\newif\ifcomments
\commentstrue
\documentclass[letterpaper,twocolumn,10pt]{article}
\usepackage[section]{placeins} 
\usepackage{esvect}
\usepackage{xspace}
\usepackage{xcolor}
\usepackage[draft]{hyperref}
\definecolor{olive}{rgb}{0.3, 0.4, .1}
\definecolor{pinegreen}{cmyk}{0.92,0,0.59,0.25}
\usepackage[textwidth=0.15\textwidth,textsize=scriptsize]{todonotes}
\usepackage{algorithm}
\usepackage{algorithmicx}
\usepackage{algpseudocode}
\usepackage{amsthm}
\usepackage{amssymb}

\usepackage{macros}
\usepackage{ioa_code_small-1}

\theoremstyle{plain}
\newtheorem{thm}{Theorem}[section]
\newtheorem{lem}[thm]{Lemma}

\newtheorem{cor}[thm]{Corollary}
\theoremstyle{definition}
\newtheorem{defn}{Definition}[section]

\usepackage{mathtools} % for "\DeclarePairedDelimiter" macro

\newcommand{\scheduler}{scheduler\xspace}
\begin{document}
\newcommand{\arcomm}[1]{\todo[color=green,bordercolor=black,linecolor=black]{\textsf{\scriptsize\linespread{1}\selectfont ARP: #1}}}
\newcommand{\arcommin}[1]{\todo[inline,color=green,bordercolor=green,linecolor=green]{\textsf{ARP: #1}}}
\newcommand{\boxedtext}[1]{\fbox{\scriptsize\bfseries\textsf{#1}}}

\newcommand{\greenremark}[2]{
   \textcolor{pinegreen}{\boxedtext{#1}
      {\small$\blacktriangleright$\emph{\textsl{#2}}$\blacktriangleleft$}
    }}

  \definecolor{burntorange}{rgb}{0.8, 0.33, 0.0}
  \definecolor{blue}{rgb}{0.0, 0.0, 0.5}
  \newcommand{\changeremark}[2]{
   \textcolor{burntorange}{\boxedtext{#1}
      {\small$\blacktriangleright$\emph{\textsl{#2}}$\blacktriangleleft$}
}}
\newcommand{\myremark}[2]{
   % \textcolor{red}{\boxedtext{#1}
   %    {\small$\blacktriangleright$\emph{\textsl{#2}}$\blacktriangleleft$}
   %  }}
      {  \lowercase{\color{blue}}\boxedtext{#1}
      {\small$\blacktriangleright$\emph{\textsl{#2}}$\blacktriangleleft$}
    }}
  \newcommand{\myremarknew}[2]{
   \textcolor{violet}{\boxedtext{#1}
      {\small$\blacktriangleright$\emph{\textsl{#2}}$\blacktriangleleft$}
}}
\newcommand{\NewRemark}[2]{
   \textcolor{violet}{
      {\small$\blacktriangleright$\emph{\textsl{#2}}$\blacktriangleleft$}
}}

\newcommand{\redremark}[2]{
   \textcolor{red}{\boxedtext{#1}
      {\small$\blacktriangleright$\emph{\textsl{#2}}$\blacktriangleleft$}
    }}
  
\ifcomments
  \newcommand\ARP[1]{\myremark{ARP}{#1}}
  \newcommand\ARPN[1]{\myremarknew{A}{#1}}
  \newcommand\vincent[1]{{\color{red}{VG: #1}}}
  \newcommand{\warning}[1]{\redremark{\fontencoding{U}\fontfamily{futs}\selectfont\char 66\relax}{#1}}
  \newcommand\NEW[1]{\NewRemark{NEW}{#1}}
  \newcommand\TODO[1]{\greenremark{TODO}{#1}}
  \newcommand\CHANGE[1]{\changeremark{CHANGE(?)}{#1}}
\else
  \newcommand\ARP[1]{}
  \newcommand\ARPN[1]{}
  \newcommand\vincent[1]{}
  \newcommand{\warning}[1]{}
  \newcommand\NEW[1]{#1}
  \newcommand\TODO[1]{}
  \newcommand\CHANGE[1]{}
\fi
\date{}

\title{Rational Agreement in the Presence of Crash Faults} %TODO Please add
% \author{Anonymous authors}{Anonymous Institution}{[email]}{[orcid]}{[funding]}

\author{{\rm Alejandro Ranchal-Pedrosa}\\
  \small{University of Sydney} \\
%\textit% name of organization (of Aff.)\\
  \small{Sydney, Australia} \\
  \small{alejandro.ranchalpedrosa@sydney.edu.au}
  \and
  {\rm Vincent Gramoli}\\
  \small{University of Sydney and EPFL}\\
  \small{Sydney, Australia} \\
  \small{vincent.gramoli@sydney.edu.au}}
\maketitle

% TODO mandatory: add short abstract of the document
\begin{abstract}
Blockchain systems need to solve consensus despite the presence of rational users and failures.
%As some faults are out of the control of players, applying fault tolerance to game theory raises interesting challenges. 
The notion of $(k,t)$-robustness has shown instrumental to list problems that cannot be solved if $k$ players are rational and $t$ players are Byzantine or act arbitrarily.
What is less clear is whether one can solve such problems if the faults are benign. 

In this paper, we bridge the gap between games that are robust against Byzantine players and games that are robust against crash players. 
Our first result is an impossibility result:
%First, we 
We show that no $(k,t)$-robust consensus protocol can solve consensus in the crash model if $k+2t\geq n$ unless there is a particular punishment strategy, called the $(k,t)$-baiting strategy.
%which is a particular punishment strategy.
%(a variant of a punishment strategy which dominates deviating for $m$ rational deviants) with respect to the protocol.
This reveals the need to introduce \emph{baiting} as the act of rewarding a colluding node when betraying its coalition, to make blockchains more secure.

Our second result is an
%Second, we draw an 
equivalence relation between crash fault tolerant games and Byzantine fault tolerant games, which raises an interesting research question on the power of baiting to solve consensus.
To this end, we show, on the one hand, that a $(k,t)$-robust consensus protocol becomes $(k+t,t)$-robust in the crash model.
We show, on the other hand, that the existence of a $(k,t)$-robust consensus protocol in the crash model that does not make use of a baiting strategy implies the existence of a $(k-t,t)$-robust consensus protocol in the Byzantine model, with the help of cryptography.
\end{abstract}
% \begin{IEEEkeywords}
% consensus, game theory, robustness, fault tolerance
% \end{IEEEkeywords}

\section{Introduction}
\label{sec:intro}

With the advent of blockchains, there is a growing interest at the frontier between distributed 
computing and game theory. As one fundamental building block of blockchains is consensus, it 
is natural to seek equilibria in which consensus is reached despite the presence of both failures and rational players. 
Moreover, as blockchains handle valuable assets over the Internet, they are typically subject to network attacks~\cite{eclipse-attack,balance-attack,attack-of-the-clone} and should tolerate unexpected delays---an assumption called partial synchrony~\cite{DLS88}---to avoid asset losses.

There are traditionally two types of failures considered in the distributed computing literature: \emph{crash failures} after which a participant stops and \emph{Byzantine failures} when participants act arbitrarily (i.e., irrationally).
This is why considering fault tolerant distributed protocols as games requires to 
cope with a mixture of up to $k$ rational players and $t$ faulty players. 
The idea of mixing rational players with faulty players has already been extensively explored
in the context of secret sharing and multi-party computation~\cite{lysyanskaya2006rationality,FKN10,DMRS11,ADGH} but rely usually on 
a trusted central authority, called a \emph{mediator}.

Recent results~\cite{abraham2019implementing} showed that mediators could 
be implemented %with asynchronous cheap talks in $\epsilon$-$(k,t)$-robust protocol 
in a fully distributed setting when $n>3(k+t)$ and $t$ players are Byzantine.
Unfortunately, this adaptation makes it impossible to devise even a consensus solution 
that is immune to a single $t=1$ Byzantine failure as it is impossible to solve consensus 
with complete asynchrony and failures~\cite{fischer1985impossibility}.
More recent results~\cite{geffner2021lower} seem to indicate that when $n\leq 4(k+t)$ there exist equilibria that 
cannot be implemented in a distributed fashion when the player behaviors are irrational or failures are Byzantine.

These impossibility results raise the question of whether one can solve consensus in a distributed fashion with $k$ rationals and $t$ failures when the failures are crashes and communication is partially synchronous~\cite{DLS88}. 
We believe this combination to be particularly relevant in the blockchain context as players are 
incentivized to steal assets by leading other players to a disagreement---also called a \emph{fork}. Such  a  situation went undetected in Bitcoin and led the attackers to ``double spend'', effectively doubling their assets.\footnote{\url{https://www.cnet.com/news/hacker-swipes-83000-from-bitcoin-mining-pools/}.}

\paragraph*{\textbf{Our result}}
In this paper, we focus instead on the partially synchronous model,
where the bound on the delay of messages is unknown~\cite{DLS88}, to
design a protocol that solves consensus among $n$ players, where up to
$t$ players can crash and $k$ are rational players that can collude
and deviate, in what we refer to as a $(k,t)$-crash-robust protocol.
To this end, we first define $(k,t)$-crash-robustness by extending
$(k,t)$-robustness~\cite{ADGH} but replacing $t$ Byzantine players by
$t$ crash players, and we consider that rational players prefer to
form a coalition and cause a disagreement than to satisfy agreement.

To the best of our knowledge, we present the first work that obtains
bounds for the robustness of agreement against coalitions of crash and
rational players in partial synchrony in a setting where rational
players prefer to disagree. This result establishes constructively a direct
relation between $(k,t)$-robust protocols and $(k',t')$-crash-robust
protocols.

More precisely, we first prove %the impossibility result that %state-of-the-art 
that no resilient-optimal crash-fault tolerant protocol can tolerate even one rational player. Specifically, we prove that it is not possible, in general, to design a protocol that implements consensus and is $(k,t)$-crash-robust for $k+2t\geq n$, unless there is a $(k,t)$-crash-baiting strategy (that is, a punishment strategy that strictly dominates deviations towards a disagreement for a number of deviating rational players) with respect to the protocol. This means that state-of-the-art crash-fault tolerant (CFT) protocols that tolerate up to less than $n/2$ crash players~\cite{DLS88} do not tolerate even one rational player in partial synchrony. 

\emph{This raises a new research question: what would be a protocol that achieves a good compromise between crash-fault tolerance and tolerance to rational players in partial synchrony?} 

To answer this question, and since the literature 
%state of the art already deepened into this problem 
tackled the problem for Byzantine faults, we demonstrate a link between Byzantine players with rational and crash players: a protocol that tolerates $s$ Byzantine players also tolerates $s$ rational players and $s$ crash players. We complete this relation by proving that a protocol that implements consensus and is $\epsilon$-$(k,t)$-robust (where $\epsilon$ accounts for the small probability of the adversary breaking the cryptography) is also $\epsilon$-$(k+t,t)$-crash-robust. Additionally, we define a $(t',t)$-immune protocol as a protocol that tolerates up to $t'$ crash faults and up to $t$ Byzantine faults, and prove that if a protocol is $\epsilon$-$(t',t)$-immune, then it is also $\epsilon$-$(t,t'+t)$-crash-robust.

Finally, we also establish a relation in the opposite direction assuming cryptography and ignoring crash-baiting strategies: a protocol that tolerates up to $k$ rational players and up to $t$ crash players without a $(k,t)$-crash-baiting strategy also tolerates up to $\min(k,t)$ Byzantine players, that is, a $(k,t)$-crash-robust protocol is also $\min(k,t)$-immune. We then prove the relation relative to robustness: if there is a protocol $\vv{\sigma}$ that is $\epsilon$-$(k,t)$-crash-robust with $k\geq t$ then we can construct an $\epsilon$-$(k-t,t)$-robust protocol $\vv{\sigma'}$, assuming cryptography and that $\vv{\sigma}$ does not implement a $(k,t)$-crash-baiting strategy, or instead that it implements an effective $(k,t)$-crash-baiting strategy that is also an effective $(k-t,t)$-baiting strategy, where effective means that playing the baiting strategy still implements consensus. We apply the analogous if $t\geq k$ to obtain that $\epsilon$-$(k,t)$-crash-robust protocols are also $\epsilon$-$(t-k,k)$-immune, excluding baiting strategies. Finally, we discuss in detail the implications of baiting and punishment strategies to this relation.

\subsection{Related work} \label{sec:relwork}

Multiple research groups studied the consensus problem by 
combining rational players with crash 
faults~\cite{clementi2017,HV20,bei2012distributed} or by
%with rational players
%Multiple works studied the problem of consensus adding rational
%players to the crash fault tolerant
%model~\cite{clementi2017,HV20,bei2012distributed}, or 
replacing
crash faults by rational
players~\cite{groce2012byzantine}. Groce et
al.~\cite{groce2012byzantine} show a protocol that solves Byzantine
agreement in synchrony even in the presence of rational coalitions as
long as the size $k$ of the coalition is such that $k<n$ and given
complete knowledge of the adversary's preferences. Nevertheless, they
consider neither crash nor Byzantine (irrational) faults, a model also used
by Ebrahimi et al.~\cite{ebrahimi2019getting}. 

Bei et
al.~\cite{bei2012distributed} extend this synchrony result to add
crash failures, obtaining that if colluding players are allowed an
extra communication round after each synchronous round, then there is
no protocol that can tolerate even 2 colluding rational players and 1
crash player under their model. Clementi et al.~\cite{clementi2017} study the problem
of fair consensus in the presence of rational and crash faults in the
synchronous gossip communication model. Fair consensus adds a new
property, \emph{fairness}, defined by all players sharing the same
probability of their proposal being decided. The gossip communication
model allows all agents to contact at most one neighbor via a
push/pull operation at every round. 

Harel et
al.~\cite{harel2020consensus} assume a set of utilities for rational
players such that they guarantee solution preference, meaning that all
rational players want to satisfy all properties of consensus. Finally,
Halpern et al.~\cite{HV20} extend the results on fair consensus in the
synchronous model. None of the results listed so far consider either the
partially synchronous model or rational players that are
interested in disagreeing in blockchains to maximize their profit.

Secret sharing and multi-party computation already explored a combined
group of rational players with faulty
players~\cite{lysyanskaya2006rationality,FKN10,DMRS11,ADGH}, specially
focused on implementing trusted mediators through
cheap-talk~\cite{ADGH}, i.e., private pairwise communication channels
of negligible cost. In particular, Abraham et al.~\cite{ADGH} showed in
2006 that there are $\epsilon$-$(k,t)$-robust protocols that implement
mediators with synchronous cheap-talk if $n>k+2t$, to later extend it
in 2019 to $\epsilon$-$(k,t)$-robust protocols that implement
mediators with asynchronous cheap-talk for $n>3(k+t)$. 

Since it is well known that it is impossible to implement a $1$-immune
protocol that solves consensus in the asynchronous
model~\cite{LSP82,fischer1985impossibility}, 
Ranchal-Pedrosa et al.~\cite{ranchal2021huntsman} devised the Huntsman
protocol, a protocol that is $\epsilon$-$(k,t)$-robust and implements
the consensus problem with $n>\max(\frac{3}{2}k+3t,2(k+t))$, which is
the highest robustness to date in the presence of Byzantine players and
coalitions of rational players that may be interested in causing
disagreement. Our result implies that the Huntsman protocol is
$\epsilon$-$(k+t,t)$-crash-robust for
$n>\max(\frac{3}{2}k+3t,2(k+t))$.

Other works have also explored the other properties of the
consensus problem (besides the agreement property). Added to the works that focused on fairness, Vila\c{c}a et
al.~\cite{Vilaca2012,vilacca2011n} and Amoussou-guenou et
al.~\cite{Amoussou-guenou2020} focus on the properties of termination and
validity in a model where communication and local computation incur a
non-negligible cost for players, and disregarding potential interests of
rational players in causing a disagreement. Abraham et al.~\cite{ADH19} explore the problem of
leader election by redefining fairness in the sense of a randomized
dictatorship in which all players have equal probability of being
elected. This protocol can also be used to solve fair consensus.

To the best of our knowledge, we present the first work that obtains
bounds for the robustness of agreement against coalitions of crash and
rational players in partial synchrony, and that establishes
constructively direct relations between $(k,t)$-robust protocols and
$(k',t')$-crash-robust protocols.
\subsection{Roadmap}
% \ARP{change}
The rest of the paper is structured as follows: Section~\ref{sec:model} presents our model and definitions taken from the literature, Section~\ref{sec:imp} shows the impossibility of resilient-optimal crash-fault tolerance in the presence of rational players. Section~\ref{sec:crashbyz} establishes the relation between Byzantine players with rational players and crash players, first by showing in Section~\ref{sec:crashbyz1} that for every Byzantine player tolerated by a consensus protocol, the same protocol can instead tolerate one rational player and one crash player, and second by showing in Section~\ref{sec:crashbyz2} that if a consensus protocol tolerates $k$ rational players and $t$ crash players then there is a consensus protocol that tolerates $\min(k,t)$ Byzantine players. % following, in Section~\ref{sec:proc} we analyse the implications of punishment and baiting strategies,
We finally conclude and detail future work in Section~\ref{sec:con}.

% , and extend to rational players and crash players the Huntsman protocol, the protocol that tolerates the greatest number of rational players and Byzantine players at the time of writing, to show the protocol that tolerates the greatest number of crash and rational players under this model
%   The rest of the paper is structured as follows: Section~\ref{sec:model} presents our computational model and some preliminary definitions, Section~\ref{sec:imp} introduces the definition of a baiting strategy and shows  %the impossibility of consensus in the presence of rational and Byzantine players without a baiting strategy, 
%   that it is impossible to solve the \problem problem without a baiting strategy,
%   in Section~\ref{sec:psync} we show bounds for a deposit and a reward to implement an effective baiting strategy, Section~\ref{sec:proc} presents the \Huntsman protocol, the first protocol that solves the \problem problem without synchrony and without solution preference. %that is $\epsilon$-$(k,t)$-robust for $n>\max(\frac{3}{2}k+3t,2(k+t))$. 
  
\section{Preliminaries}
\label{sec:model}
%\TODO{make sure not redundant}
%\TODO{define deposit assumption}\vincent{We assume that players can check that others have deposited. this is a bit protocol-specific.}
% We consider a partially synchronous communication network, in which
% messages can be delayed by up to a bound 
% %$\Lambda$ 
% that is unknown.

%\vspace{-0.1em}
Our focus is on a partially synchronous communication network~\cite{DLS88}, where
messages can be delayed by up to a bound that is unknown, but not
indefinitely.
For this purpose, we adapt the synchronous and asynchronous models of
Abraham et al.~\cite{ADGH,abraham2019implementing} to partial
synchrony, including the definitions of Halpern et al.~\cite{Halpern2020} to introduce crash players, with the appropriate modifications to account for partial synchrony~\cite{ranchal2021huntsman}. %  In order to adapt these models to partial synchrony, we adopt the modifications illustrated by \ARP{remove?}Ranchal-Pedrosa
% et al.\ARP{cite?}
Hence, our model consists of a game played by a set $N$ of
players, with $|N|=n$. The players in $N$ can be of four different types:
correct, rational, crash or Byzantine. 

In order to model partial synchrony, we introduce the \scheduler as an additional player that will model the delay on messages derived from partial synchrony.
%The set of players playing the game is $N$ of
%size $|N|=n$. 
The game is in \emph{extensive form}, described by a game tree whose
leaves are labeled by the utility $u_i$ of each player $i$. We assume that
players alternate making moves with the \emph{\scheduler}: first the
\scheduler moves, then a player moves, then the \scheduler moves and
so on. The scheduler's move consists of choosing a player $i$ to move
next and a set of messages in transit targeting $i$ that will be delivered
just before $i$ moves (so that $i$'s move can depend on all the
messages $i$ receives). Every non-leaf tree node is associated with either
a player or the \scheduler. The \scheduler is bound to two
constraints. First, the \scheduler can choose to delay any message up
to a bound, known only to the \scheduler, before which he must have
chosen all receivers of the message to move and provided them with
this message, so that they deliver it before making a move. Second,
the \scheduler must eventually choose all players that are still
playing. That is, if player $i$ is playing at time $e$, then the
\scheduler chooses him to play at time $e'\geq e$.

Each player $i$ has some \textit{local state} at each node, which
translates into the initial information known by $i$, the messages $i$
sent and received at the time that $i$ moves, and the moves that $i$
has made. The tree nodes where a player $i$ moves are further partitioned
into \textit{information sets}, which are sets of nodes
in the game tree that contain the same local state for the same
player, in that such player cannot distinguish them. We assume that
the \scheduler has complete information, so that the \scheduler's
information sets simply consist of singletons.

Since faulty or rational players can decide not to move during their
turn, we assume that players that decide not to play will at least
play the \textit{default-move}, which consists of notifying to the
\scheduler that this player will not move, so that the game continues
with the \scheduler choosing the next player to move. Thus, in every node where the scheduler is to play a move, the \scheduler can play any move that combines a player and a subset of messages that such player can deliver before playing.
 Then,
the selected player moves, after which the \scheduler selects again
the next player for the next node, and the messages it receives, and
so on. The \scheduler thus alternates with one player at each node
down a path in the game tree up to a leaf. A \textit{run} of the game
is then a path in the tree from the root to a leaf. 
%Each player has a type taken from the type space
%$\mathcal{T}=\{\text{Byzantine, rational, correct}\}$.

\paragraph*{Strategies}  % A player $i$ makes a move following an action. 
We denote the set of actions of a player $i$ (or the \scheduler) as
$A_i$, and a strategy $\sigma_i$ for that set of actions is denoted as
a function from $i$'s information sets to a distribution over the
actions. 
% For example, if the strategy $\sigma_i$ for player $i$ is to
% follow a protocol regardless of the information set, then $i$ will make multiple moves that correspond
% to following such protocol $\sigma_i$, such as sending multiple
% messages, or waiting to deliver certain messages from other players
% (i.e., $i$ moves following his strategy $\sigma_i$).

We denote the set of all possible strategies of player $i$ as
$\mathcal{S}_i$. Let $\mathcal{S}_I=\Pi _{i\in I} \mathcal{S}_i$ and
$A_I=\Pi_{i\in I} A_i$ for a subset $I\subseteq N$. Let
$\mathcal{S}=\mathcal{S}_N$ with $A_{-I}=\Pi _{i\not\in I} A_i$ and
$\mathcal{S}_{-I}=\Pi _{i\not\in I} \mathcal{S}_i$.  A \textit{joint
strategy} $\vv{\sigma}=(\sigma_0,\sigma_1,...,\sigma_{n-1})$ draws
thus a distribution over paths in the game tree (given the scheduler's strategy $\sigma_s$), where $u_i(\vv{\sigma},\sigma_s)$ is player's $i$ expected utility if $\vv{\sigma}$ is played along with a strategy for the scheduler $\sigma_s$. A strategy $\theta_i$ \textit{strictly dominates} $\tau_i$ for $i$ if for all $\vv{\phi}_{-i}\in \mathcal{S}_{-i}$ and all strategies $\sigma_s$ of the \scheduler we have $u_i(\theta_i,\vv{\phi}_{-i},\sigma_s)>u_i(\tau_i,\vv{\phi}_{-i},\sigma_s)$.

Given some desired functionality
$\mathcal{F}$, a protocol is the recommended joint strategy
$\vv{\sigma}$ whose outcome satisfies $\mathcal{F}$, and an associated
game $\Gamma$ for that protocol is defined as all possible deviations
from the protocol~\cite{ADGH}. In this case, we say that the protocol
$\vv{\sigma}$ \textit{implements} the desired functionality. Note that
both the \scheduler and the players can use probabilistic strategies.
\paragraph*{Failure model}
$k$ players out of $n$ can be rational and up to $t$ of them can be
faulty (i.e. Byzantine or crash), while the rest are correct. Correct players follow the
protocol: the expected utility of correct player $i$ is equal and
positive for any run in which $i$ follows the protocol, and $0$ for
any other run. Rational players can deviate to follow the strategy
that yields them the highest expected utility at any time they are to
move, while Byzantine players can deviate in any way, even not
replying at all (apart from notifying the \scheduler that they will
not move). A crash player $i$ behaves exactly as a correct player,
except that it can crash in any round of any run. If $i$ crashes in
round $m$ of run $r$, then it may send a message to some subset of
agents in round $m$, but from then on, it sends no further
messages (except for playing the default-move). We will detail further the utilities of rational players
after defining the Byzantine consensus problem.

We let rational players in a coalition and Byzantine
players (in or outside the coalition) know the types of all players,
so that these players know which players are the other faulty,
rational and correct players, while
the rest of the players only know the upper bounds on the number of
rational and faulty players, i.e., $k$ and $t$ respectively, and
their own individual type (that is, whether they are rational,
Byzantine, crash or correct). 
% purpose, we assume the worst-case for our problem, in most cases
% being: rational players in a coalition and Byzantine players (in or
% outside the coalition) know the types of all players, while the rest
% of the players only know the upper bounds on the number of rational, and either Byzantine or crash players, i.e., $k$ and $f$ respectively.
% The market value of % the system is $\mathcal{V}.$

\paragraph*{Cheap talks}
As we are in a fully distributed system, without a trusted central
entity like a mediator, we assume \textit{cheap-talks}, which are
private pairwise communication channels. We also assume negligible
communication cost through these channels. Rational and correct players are
not interested in the number of messages exchanged. Similarly, we
assume the cost of performing local computations (such as validating
proposals, or verifying signatures) to be negligible.

\paragraph*{Cryptography}
We require the use of cryptography, for which we reuse the assumptions
of Goldreich et al.~\cite{goldreichplay}: polynomially bounded players
and the enhanced trapdoor permutations. 
%We further assume the
%existence of a public-key infrastructure (PKI). 
In practice, these two
assumptions mean that players can sign unforgeable messages, and that
they can perform oblivious transfer.

\paragraph*{Robustness}
We restate Abraham's et al.~\cite{ADGH} definitions of $t$-immunity, $\epsilon$-$(k,t)$-robustness and the most recent definition of $k$-resilient
equilibrium~\cite{abraham2019implementing} to consider multiplayer deviations. We also add the definitions of $t$-crash-immune and $(k,t)$-crash-robust equilibrium. Notice that our definition of $(k,t)$-crash-robust equilibrium differs from Bei et al's $(c,t)$-resilient equilibrium, since we simply extend Abraham et al's~\cite{ADGH} definition of a $(k,t)$-robust equilibrium. While this definition is too restrictive for Bei et al's model, it fits properly our illustration of the utilities of the rational players that we define in this work.
%A join strategy is $t$-immune if it tolerates coalitions of up to $t$
%Byzantine faults and $k$-resilient if it tolerates coalitions of
%up to $k$ rational players. \vincent{this is a bit too simplistic as a definition.}

The notion of $k$-resilience is motivated in distributed computing by
the need to tolerate a coalition of $k$ rational players that can all coordinate
actions. A joint strategy is $k$-resilient if no coalition of size $k$
can gain greater utility by deviating in a coordinated way. 
%\TODO{do we need resilience here?}
\begin{defn}[$k$-resilient equilibrium]
  A joint strategy $\vv{\sigma}\in \mathcal{S}$ is a \textit{$k$-resilient equilibrium} (resp. \textit{strongly k-resilient equilibrium}) if, for all sets $K$ of rational players such that $K\subseteq N$ with $|K|\leq k$, all $\vv{\tau}_{K}\in \mathcal{S}_K$, all strategies $\sigma_s$ of the \scheduler, and for some (resp. all) $i\in K$ we have:
  % \begin{equation*}
  $u_i(\vv{\sigma}_K, \vv{\sigma}_{-K},\sigma_s) \geq u_i(\vv{\tau}_{K}, \vv{\sigma}_{-K},\sigma_s).$
    % \end{equation*}
\end{defn}

The notion of $t$-immunity is motivated by
the need to tolerate $t$ faulty players. An equilibrium is
$t$-immune if the expected utility of the non-faulty players is not affected by the
actions of up to $t$ other faulty players. The $\epsilon$ here accounts for the (small) probability
of the coalition breaking cryptography, as previously assumed in the literature~\cite{ADGH}:

\begin{defn}[$\epsilon$-$t$-immunity]
  A joint strategy $\vv{\sigma}\in \mathcal{S}$  is $\epsilon$-$t$-immune if, for all sets $T$ of Byzantine players such that $T\subseteq N$ with $|T|\leq t$, all $\vv{\tau}\in \mathcal{S}_T$, all strategies $\sigma_s$ of the \scheduler, and all $i\not \in T$, we have:
  % \begin{equation*}
  $u_i(\vv{\sigma}_{-T}, \vv{\tau}_T,\sigma_s) \geq u_i(\vv{\sigma},\sigma_s)-\epsilon.$
    % \end{equation*}
\end{defn}

% We speak of $t$-byzantine-immunity for immunity against Byzantine faults and $t$-crash-immunity for immunity against crash faults. If we refer to $t$-immunity, we assume Byzantine faults.

A joint strategy $\vv{\sigma}$ is an $\epsilon$-$(k,t)$-robust
equilibrium if no coalition of $k$ rational players can coordinate to
increase their expected utility by $\epsilon$ regardless of the
arbitrary behavior of up to $t$ faulty players, even if the
faulty players join their coalition. % For the case of an
% $\epsilon$-$(k,t)$-crash-robust equilibrium we take the analogous
% definition replacing the up to $t$ Byzantine players by up to $t$
% crash players.

%   \begin{defn}[$\epsilon$-$(k,t)$-robust and $\epsilon$-$(k,t)$-resistant equilibria]
%   A joint strategy $\vv{\sigma}\in \mathcal{S}$  is an $\epsilon$-$(k,t)$-robust (resp. $\epsilon$-$(k,t)$-resistant) equilibrium
%   if for all $K, T\subseteq N$ such that $K$ only contains rational players, $T$ only contains Byzantine (resp. crash) players, \vincent{We should simplify by removing unnecessary information:}$K\cap T = \emptyset, |K|\leq k,$ and $|T|\leq t$, for all $\vv{\tau}_T \in \mathcal{S}_T$, for all $\vv{\phi}_K\in\mathcal{S}_K$, for some % \vincent{What does "all (resp. all)" mean?}
%   $i\in K$, and all strategies of the \scheduler $\sigma_s$, we have $u_i(\vv{\sigma}_{-T},\vv{\tau}_T,\sigma_s)\geq u_i(\vv{\sigma}_{N-(K\cup T)}, \vv{\phi}_K, \vv{\tau}_T,\sigma_s)-\epsilon$. We speak instead of a $(k,t)$-robust (resp. $(k,t)$-resistant) equilibrium if $\epsilon=0$.
% \end{defn}

\begin{defn}[$\epsilon$-$(k,t)$-robust equilibrium]
  Let $K$ denote the set of $|K|=k$ rational players and let $T$ denote the set of $|T|\leq t$ Byzantine players, $K\cap T = \emptyset$. A joint strategy $\vv{\sigma}\in \mathcal{S}$ is an $\epsilon$-$(k,t)$-robust equilibrium if for all $K, T\subseteq N$, for all $\vv{\tau}_T \in \mathcal{S}_T$, for all $\vv{\phi}_K\in\mathcal{S}_K$, and all strategies of the \scheduler $\sigma_s$, there exists $i\in K$ such that:
  \begin{equation*}
    u_i(\vv{\sigma}_{-T},\vv{\tau}_T,\sigma_s)\geq u_i(\vv{\sigma}_{N-(K\cup T)}, \vv{\phi}_K, \vv{\tau}_T,\sigma_s)-\epsilon.
    \end{equation*}
\end{defn}
%We say that a $\epsilon$-$(k,t)$-robust equilibrium is \emph{$\epsilon$-$(k,t)$-crash-robust} equilibrium if the failures are crashes.
If $\epsilon=0$, we simply refer to an $\epsilon$-$(k,t)$-robust equilibrium as a 
$(k,t)$-robust equilibrium, and an $\epsilon$-$t$-immune protocol as a $t$-immune protocol. % Also, we speak of $\epsilon$-$(k,t)$-byzantine-robustness and $\epsilon$-$(k,t)$-byzantine-robustness to refer to Byzantine or crash faults, respectively. Again, if not specified, $\epsilon$-$(k,t)$-robustness refers to Byzantine faults

  Given some game $\Gamma$ and desired functionality $\mathcal{F}$, we
say that a protocol $\vv{\sigma}$ is a $k$-resilient protocol for
$\mathcal{F}$ if $\vv{\sigma}$ implements $\mathcal{F}$ and is a $k$-resilient equilibrium. %   Given some game $\Gamma$ and desired functionality $\mathcal{F}$, we
% say that a protocol $\vv{\sigma}$ is a $k$-resilient equilibrium for
% $\mathcal{F}$ if $\vv{\sigma}$ is a $k$-resilient equilibrium and the outcome of $\vv{\sigma}$ satisfies
% $\mathcal{F}$. We also say in such a case that the protocol
% \textit{implements} the functionality $\mathcal{F}$. % A mechanism is
% \textit{practical} if it survives iterated deletion of weakly
% dominated strategies, as defined in previous
% work~\cite{rabin1989verifiable}.
% All possibility results in this paper
% are practical mechanisms.
% For example, if $\vv{\sigma}$ is a
% k-resilient protocol for the consensus problem, then in all runs of
% $\vv{\sigma}$, every correct player terminates and agrees on the same
% valid value despite there being up to $k$ rational players. 
We extend this notation to $t$-immunity and
$\epsilon$-$(k,t)$-robustness. % We speak of $t$-immune,
% $k$-resilient and $\epsilon$-$(k,t)$-robust protocols for a
% functionality to refer to $t$-immunity, $k$-resilience and
% $(k,t)$-robustness of the associated .
The required functionality of this paper is thus reaching consensus.
  
\paragraph*{Punishment and baiting strategy}
 We also restate the definition of a punishment strategy~\cite{ADGH}
as a threat from correct and rational players in
order to prevent other rational players from deviating. For example,
in a society, this threat can be viewed as a punishment strategy of the judicial system against
committing a crime. Slashing a player if he is found to be guilty of
deviating is another punishment strategy.

  % In the following we need to consider a set of players that will act against the coalition by bidding. %, this makes the size of $P$ a parameter.\vincent{Is this up-to-date?}
The punishment strategy guarantees that if $k$ rational players deviate, then $t+1$ players can lower the utility of these rational players by playing the punishment strategy. 
% \vincent{TODO: change this def into its $\epsilon$ variant and if the definition is identical to previous work, then it is better to place it in the previous section.}
% We further define an $\epsilon$-$(k,t)$-crash-punishment strategy for the case where the $t$ players are crash and not Byzantine players 

%  \begin{defn}[$(k, t)$-punishment strategy]  A joint strategy $\vv{\rho}$ is a $(k,t)$-punishment strategy with respect to $\vv{\sigma}$ if for all $K,T,P\subseteq N $ such that $K,T,P$ are disjoint, members of $K$ are rational players, members of $T$ are Byzantine players, and the rest are correct, $|K|\leq k,|T|\leq t,|P|> t$, for all $\vv{\tau}\in \mathcal{S}_T$, for all $\vv{\phi}_K\in \mathcal{S}_K$, for all $i\in K$, and all strategies of the \scheduler $\sigma_s$, we have:
  \begin{defn}[$(k, t)$-punishment strategy]  Let $K,T,P\subseteq N$ be disjoint sets of rational players, Byzantine players and correct players, respectively, such that $|K|\leq k,|T|\leq t,|P|> t$.
  A joint strategy $\vv{\rho}$ is a $(k,t)$-punishment strategy with respect to $\vv{\sigma}$ if 
  %for all $K,T,P\subseteq N $ such that $K,T,P$ are disjoint, members of $K$ are rational players, members of $T$ are Byzantine players, and the rest are correct, $|K|\leq k,|T|\leq t,|P|> t$, for all $\vv{\tau}\in \mathcal{S}_T$, 
  for all $\vv{\phi}_K\in \mathcal{S}_K$, for all $i\in K$, and all strategies of the \scheduler $\sigma_s$, we have:\vspace{-1em}

    \begin{equation*}
      u_i(\vv{\sigma}_{-T},\vv{\tau}_T,\sigma_s) > u_i(\vv{\sigma}_{N-(K\cup T \cup P)}, \vv{\phi}_K, \vv{\tau}_T, \vv{\rho}_P,\sigma_s)
    \end{equation*}
  \end{defn}
  Intuitively, a punishment strategy represents a threat to prevent rational players from deviating, in that if they deviate, then players in $P$ can play the punishment strategy $\vv{\rho}$, 
  which decreases the utility of rational players with respect to following the strategy $\vv{\sigma}$. 
  %and the deviating rational players decrease their utility with respect to following $\vv{\sigma}$. % If the punishment strategy does not make use of cryptography, then

  Recent works consider a particular type of punishment strategies, defined as baiting strategies~\cite{ranchal2021huntsman}. In baiting strategies, some players that enforce the punishment on the deviants are rational players from within the coalition: they have an incentive to bait other players into a punishment strategy. An example of a baiting strategy can be found when law-enforcement officers offer an economic reward, or a reduced sentence, if a member of a criminal group helps them arrest the rest of the group.

  \begin{defn}[$(k, t, m)$-baiting strategy]\label{def:ekf-bs}
    Let $K,T\subseteq N$ be disjoint sets of rational players and Byzantine players, respectively. Let $P\subseteq N$ be a set of baiters, composed of rational and correct players, and let the rest of correct players be $C=N-(K\cup T \cup P)$. A joint strategy $\vv{\eta}$ is a $(k,t,m)$-baiting strategy with respect to a strategy $\vv{\sigma}$ if $\vv{\eta}$ is a $(k-m,t)$-punishment strategy with respect to $\vv{\sigma}$, with $0< m \leq k$, $|K|\leq k,|T|\leq t, |P|> t,|P\cap K|\geq m$, for all $\vv{\tau}\in \mathcal{S}_T$, all $\vv{\phi}_{K\backslash P}\in \mathcal{S}_{K\backslash P}-\{\vv{\sigma}_K\}$, all $\vv{\theta}_{P}\in \mathcal{S}_{P}$, all $i\in P$, and all strategies of the \scheduler $\sigma_s$, we have:
    \begin{equation*}
      u_i(\vv{\sigma}_{C}, \vv{\phi}_{K\backslash P}, \vv{\tau}_T, \vv{\eta}_P, \sigma_s)\geq       u_i(\vv{\sigma}_{C}, \vv{\phi}_{K\backslash P},\vv{\tau}_T,\vv{\theta}_{P},\sigma_s)
    \end{equation*}
    % while for all $i\in K$:
    % \begin{equation*}
    %   u_i(\vv{\sigma}_{N-(K\cup T \cup P)}, \vv{\phi}_K, \vv{\tau}_T, \vv{\eta}_P)<       u_i(\vv{\sigma}_{-T}, \vv{\tau}_T)
    % \end{equation*}
    Additionally, we call this strategy a strong $(k,t)$-\textit{baiting strategy} in the particular case where for all rational coalitions $K \subseteq N$ such that $|K|\leq k+f$, $|K\cap P|\geq m$ and all $\vv{\phi}_{K\backslash P}\in \mathcal{S}_{K\backslash P}$, we have:
    \begin{equation*}
      \sum_{i\in K}u_i(\vv{\sigma}_{N-(K\cup P)}, \vv{\phi}_{K\backslash P}, \vv{\eta}_P,\sigma_s) \leq \sum_{i\in K}u_i(\vv{\sigma}, \sigma_s).
    \end{equation*}
    We write (strong) $(k,t)$-baiting strategy instead to refer to a (strong) $(k,t,m)$-baiting strategy for some $m$, with $0< m \leq k$.
    \end{defn}

    We also refer to an \textit{effective} baiting strategy if playing such strategy still implements the desired functionality. An example of an effective baiting strategy for the problem of consensus is rewarding a baiter for exposing a disagreement attempt before it takes place, if it is resolved~\cite{ranchal2021huntsman}. Strong baiting strategies are strategies in which the coalitions formed entirely by rational players end up collectively losing compared to if they had just followed the protocol, even if a subset of them play the baiting strategy. This prevents coalitions from baiting themselves just for the purpose of maximizing the sum of their
    %getting a greater sum of 
    utilities. % Following the same example, if the reward from baiting is significantly larger than the collective reward from causing a disagreement, the coalition may prefer one of the members of the coalition to actually bait the rest just to get the reward (which they can split after). This is not a scenario that deviating players prefer if the baiting strategy is a strong baiting strategy.

    We extend the above-defined terms to their analogous crash fault tolerant counterparts by replacing Byzantine players by crash players in all their definitions, in what we refer to as $(t)$-crash-immunity, $(k,t)$-crash-robustness, $(k,t)$-crash-punishment strategy and $(k,t,m)$-crash-baiting strategy.
    % We abuse notation in all definitions where faulty players are involved, such as $t$-immunity, we refer to the fact that these faults are crash faults by adding the prefix to the name, i.e. $t$-crash-immunity. If no prefix is added, we refer to Byzantine faults instead. Also, note that the rest of the players that are neither faulty nor rational are correct in all of these definitions.

 \paragraph*{Consensus}
We recall the Byzantine consensus problem~\cite{LSP82} in the presence of rational players:
The \emph{Byzantine consensus problem} is, given $n$ players, each with an initial value, to ensure (1) \emph{Agreement}, in that no two non-deviating players decide different values, (2) \emph{Validity} in that if a non-deviating player decides a value $v$, then $v$ has been proposed by some player, and (3) \emph{Termination} in that all non-deviating players eventually decide.

% \begin{itemize}
% \item{\emph{Agreement}}: No two non-deviating players decide different values,
% \item{\emph{Validity}}: If a non-deviating player decides a value $v$, then $v$ has been proposed by some player,
% \item{\emph{Termination}}: All non-deviating players eventually decide.
% \end{itemize}

%   Finally, we assume that if a %$k$ rational players in a deviating
% coalition manages to cause a disagreement, then it obtains a payoff of at most $\mathcal{G}$, which
% we call the \textit{total gain}. This total gain may be, for example, the entire
% market value of the system. We assume that a coalition with $k$ rational players
% split the total gain into $k$ parts, which we call the $\textit{gain}$ $g$ of a rational player. We assume that this value $g$ is
% always greater than agreeing as long as $k<n$, meaning that rational players always prefer to cause a disagreement colluding with a subset of the players, if that was possible. 
\paragraph*{Disagreements} Notice a disagreement of consensus can mean two or more disjoint
groups of non-deviating players deciding two or more 
conflicting decisions~\cite{singh2009zeno}. We speak of the
\textit{disagreeing} strategy as the strategy in which deviating
players collude to produce a disagreement, and of a coalition
\textit{disagreeing} to refer to a coalition that plays the
disagreeing strategy.  % Inw order to distinguish the value $t$
% from $t$-immunity and from $(k,t)$-robustness (or $(k,t)$-crash-robust), we set
% $t=\ceil{\frac{n}{3}}-1$ in the rest of this paper. We will make
% use of $f$ instead to refer to more restrictive bounds for Byzantine
% players, i.e., $f\leq t$.
  % \subsection{Synchronous \gametype game }
  % \label{sec:syncgame}
 % We conjecture that this definition of the validity property would also yield better resu
 % \ARP{the trick here is that, in synchrony, if a player arrives after the bound and claims anything regarding the decision, everyone else can ignore him because he did not arrive on time to claim anything. Nobody believes that claim. That is why synchrony is a big assumption.}

 \paragraph*{Rational players} The utilities of a rational player $i$ depending on its actions and a particular run of the protocol are as follows:
    \begin{enumerate}
      \item \label{str:cor} If an agreement is reached, then $i$ gets utility
$u_i(\vv{\sigma}_{-T},\vv{\tau}_T)\geq \epsilon$ where $\epsilon>0$.
      \item \label{str:dis} If coalition where $i$ is a member successfully performs a disagreement, then $i$ gets utility $u_i(\vv{\sigma}_{N-K-T},\vv{\phi}_{K\cup T})=g>\epsilon$.
      \item \label{str:ben}If the protocol does not terminate, then $i$ obtains a negative utility.
      \item \label{str:ter} If player $i$ suffers a disagreement caused by a coalition external to $i$, then $i$ obtains a negative utility.
      \end{enumerate}

      %Note that this is not a whimsical choice of utilities. 
      Note that we do not consider fairness of consensus, computational costs, communication costs, or a preference from a proposal over another, but we rather focus on the interests in causing a disagreement. As such, setting the expected utility of causing agreement to be greater than that of causing disagreement would inevitably lead to rational players behaving exactly as correct players. Similarly, selecting the utilities of not terminating greater than the utilities of causing agreement would lead to all rational players behaving as crash players. For both of these cases, the state-of-the-art bounds are applicable. The set of utilities that we choose here also reflects a realistic behavior of rational players in the blockchain context, where players can get an economic incentive from a fork (as it is the case when they double spend by forking).
      It is easy to notice that not terminating as well as deciding a proposal while a coalition causes a disagreement are strictly dominated by reaching agreement, which is in turn strictly dominated by causing a disagreement.
     A baiting strategy means %a possible behavior by a rational player, in which he
     a rational player possibly 
     %initially joins 
     joining 
     the coalition only to later betray it in exchange of a reward. This additional strategy would strictly dominate the strategy to disagree by definition. %of baiting strategy.
% \begin{defn}[\problem]
%   Consider a system with $n$ players, a protocol $\vv{\sigma}$ solves the \problem problem if it implements consensus, and is $t$-immune and $\epsilon$-$(k,f)$-robust for some $k, f>0$ such that $n \leq 3(k+f)$.
% \end{defn}

% \TODO{change here definitions to reflect also crash faults not Byzantine behaviour}
\section{Impossibility result}
\label{sec:imp}

In this section, we show that resilient-optimal protocols cannot tolerate even one rational player. 
% \TODO{talk about previous result on crash fault tolerance, and how we prove here that one cannot be crash resilient optimal and at the same time support any rational at all}
Previous results %on crash fault tolerance 
showed that a resilient-optimal protocol 
%in the presence of crash faults 
tolerates up to $t<n/2$ crash faults~\cite{DLS88}. We show in Lemma~\ref{lem:imp1} and Theorem~\ref{thm:imp1} that this number of crash faults does not tolerate even one rational player. 
\begin{lem}
  % It is impossible to obtain a protocol $\vv{\sigma}$ that implements consensus, is $t$-crash-immune and $k$-resilient for $k+2t\geq n$. \ARP{excluding baiting strategies}
  Let $\vv{\sigma}$ be a protocol that implements consensus such that there is no $(k,t,m)$-crash-baiting strategy with respect to $\vv{\sigma}$. Then, it is impossible for $\vv{\sigma}$ to be $t$-crash-immune and $k$-resilient for $k+2t\geq n,\, m>\frac{k-n}{2}+t$.
  \label{lem:imp1}
\end{lem}
\begin{proof}
% \vincent{The definition of f-crash-immune should come before the theorem statement.}  \ARP{it is in the previous section.} 
If a protocol is $t$-crash-immune, that means that the protocol must
terminate even if $t$ players do not participate in it at all, since
$t$ players may have crashed from the beginning. Therefore, the
protocol must terminate and decide with the participation of $n-t$
players.

  Since the protocol must be able to terminate with the participation
of at most $n-t$ players,
consider now that there are no crash players and there are exactly $k$
rational players. Let us find a disjoint partition of correct players $A$ and $B$
such that $k+|A|+|B|=n$. If $k+|A|\geq n-t$ and $k+|B|\geq n-t$ then
the $k$ rational players can cause a disagreement, which can occur if
$k+2t\geq n$. There is only left to prove that $k$ rational players
will try to cause a disagreement. For this, let us consider the
minimum number of rational players $m$ that must not try to cause a
disagreement for the remaining deviating rational players to not be
able to cause a disagreement. That is, for which values we have
$k-m+|A|<n-t$ and $k-m+|B|<n-t$, hence resulting in
$m>\frac{k-n}{2}+t$. 
%% TODO rephrase:
Therefore, the utilities for rational players from causing a
disagreement are greater than from causing agreement. This means that
at least enough rational players will deviate and cause the
disagreement unless there is a $(k,t,m)$-baiting strategy that
prevents $m$ rational players from deviating into a disagreement.
\end{proof}

The next theorem follows directly from Lemma~\ref{lem:imp1} because every $(k,t)$-crash-robust protocol must also be $t$-crash-immune and $k$-resilient.
\begin{thm}
  Let $\vv{\sigma}$ be a protocol that implements consensus such that there is no $(k,t)$-crash-baiting strategy with respect to $\vv{\sigma}$. Then, it is impossible for $\vv{\sigma}$ to be $(k,t)$-crash-robust for $k+2t\geq n$.
  \label{thm:imp1}
\end{thm}
% \begin{proof}
%   The proof is analogous to that of Lemma~\ref{lem:imp1}, since by definition every $(k,t)$-crash-robust protocol must also be $t$-crash-immune and $k$-resilient.
%   \end{proof}
\begin{cor}
  Let $\vv{\sigma}$ be a protocol that implements consensus such that there is no $(1,t)$-crash-baiting strategy with respect to $\vv{\sigma}$ and is $t$-crash-immune for $t<n/2$. Then, $\vv{\sigma}$ is not $1$-resilient.
  \label{cor2:imp1}
\end{cor}
The results from Lemma~\ref{lem:imp1}, Theorem~\ref{thm:imp1}
and Corollary~\ref{cor2:imp1} show that it is necessary to consider new bounds
for CFT protocols in terms of their crash-fault
tolerance, since their resilient-optimal bounds make them vulnerable to even one rational player for state-of-the-art protocols, because they do not consider offering a crash-baiting strategy. In Section~\ref{sec:crashbyz}, we explore the link between crash-robustness and immunity, so as to obtain results for this model with the existing protocols designed for Byzantine faults.

% \ARP{maybe discuss more, specify that we are interested in a protocol that, setting $t$, obtains the newly resilient-optimal result of $k<n-2t$. Discuss interesting values of such equation: i.e. $k+t<2n/3$, $t=0$ and $k=n-1$, etc.}
% \TODO{Corollary: the maximum value $k+t$ such that a protocol can be $(k,t)$-crash-robust is $t<n/3$ and $k<n/3$}

\section{Bridging the gap: crash and rational players as Byzantine players}
\label{sec:crashbyz}
In this section we bridge the gap between games that are robust against Byzantine players and games that are robust against crash players.
%Previous works already captured the notions of coalitions of rational players and faulty players. In this section, we bridge the gap between Byzantine players and crash players in the presence of rational players.
\subsection{From Byzantine to crash players}
\label{sec:crashbyz1} It is immediate that a protocol that tolerates up to $t$ Byzantine faults also tolerates up to $t$ crash faults, the question lies with the inclusion of rational players. We propose in Lemma~\ref{lem:rel1} a first relation between Byzantine fault tolerance and crash-fault tolerance in the presence of rational players. 
% \vincent{If we restrict the failure model to crashes, then any $t$-immune protocol is  $(t,t)$-robust. If the failures are byzantine, then we know we can find a $\epsilon$-$(k-t,t)$-robust protocol.}\ARP{not sure what this means.}
\begin{lem}
  Let $\vv{\sigma}$ be a protocol that implements consensus and is $\epsilon$-$s$-immune. Then $\vv{\sigma}$ is also $\epsilon$-$(s,s)$-crash-robust.
  \label{lem:rel1}
  \end{lem}
  \begin{proof}
    We prove this by contradiction. Let $r$ be the minimum number of players that must participate in the protocol for it to terminate, it is clear that $r\leq N-s$ since the protocol is $\epsilon$-$s$-immune. As such, let $A$ and $B$ be two disjoint sets of correct players. Since the protocol must also guarantee agreement, it follows by contradiction that if agreement was not satisfied then $|A|+s\geq r$ and $|B|+s\geq r$, but this is only possible if $r\leq \frac{n+s}{2}$. Therefore, we have $N-s \geq r> \frac{n+s}{2}$.

    We define $s_c\leq s$ and $s_k\leq s$ as the maximum combined values of tolerated crash and rational faults, respectively. It is immediate that termination is guaranteed, since rational players will participate and $r\leq N-s \leq N-s_c$. For agreement, we have $|A|+|B|+s_c+s_k<n$, with $s_c\leq s$ and $s_k\leq s$. We consider that $s_c$ crash players crash after having sent some messages only to players in $|A|$ and the $s_k$ rational players, which is the best-case for the deviating coalition (otherwise they crash sending the same message to the entire set of correct players and thus they do not contribute to disagreeing). For the $s_k$ rational players to lead players in $B$ to a different decision than the decision of players in $A$ plus the crash players $s_c$, both $s_k+s_c+|A|\geq r$ and $s_k+|B|\geq r$ must hold. This means that $|A|+|B|+2s_k+s_c\geq 2r\iff n+s_k\geq 2r \iff \frac{n+s}{2}\geq r$, however, this is a contradiction: we already showed that $\frac{n+s}{2}< r$ for $\vv{\sigma}$ to be $\epsilon$-$s$-immune. It follows that if a protocol that implements consensus is $\epsilon$-$s$-immune then it is also $\epsilon$-$(s,s)$-crash-robust.% \ARP{can detail better proof but evth here already}
  \end{proof}
  Notice that the statement of Lemma~\ref{lem:rel1} does not require to assume cryptography, and thus the same result takes place by considering $\epsilon=0$, i.e., $(k,t)$-robustness. The same occurs with Theorem~\ref{thm:rel1}. Lemma~\ref{lem:rel1} establishes a surprising yet meaningful relation between $t$-immunity and $(k,t)$-crash-robustness, further extended by Theorem~\ref{thm:rel1}: if a protocol is $(k,t)$-robust then it is also $(k+t,t)$-crash-robust. We omit the proofs of Theorems~\ref{thm:rel1} and~\ref{thm2:rel1} as they are analogous to that of Lemma~\ref{lem:rel1}.
  \begin{thm}
    \label{thm:rel1}
    Let $\vv{\sigma}$ be a protocol that implements consensus and is $\epsilon$-$(k,t)$-robust. Then, $\vv{\sigma}$ is also $\epsilon$-$(k+t,t)$-crash-robust.% \vincent{of if the failures are byzantine then it  $(k,s)$-robust, if the failures are crashes then it is  $(k+s,s)$-robust.}
  \end{thm}
  % \begin{proof}
  %   The proof is analogous to that of Lemma~\ref{lem:rel1}.
  % \end{proof}
  By Lemma~\ref{lem:rel1} and Theorem~\ref{thm:rel1}, it is possible to take existing protocols, bounds and other results that apply to immunity and robustness and apply them directly to crash-immunity and crash-robustness. % For example, we show in Corollary~\ref{cor2:rel1} that, thanks to Theorem~\ref{thm:rel1}, \ARP{remove?}Ranchal-Pedrosa et al.'s[cite] Huntsman protocol that achieves $(k,t)$-robustness for $n<\max(\frac{3}{2}k+3t,2(k+t))$ is also $(k+t,t)$-crash-robust for the same values of $k$ and $t$. 
  Moreover, Lemma~\ref{lem:rel1} establishes a parametrizable hierarchy between crash faults and Byzantine faults: for every Byzantine fault tolerated by a protocol that solves consensus, the same protocol tolerates instead one crash fault and one rational player. % deceitful fault~\cite{ranchal2020blockchain} (i.e. equivocating fault), since rational players may behave as deceitful faults in that they try to cause a disagreement, and in fact the proof of Lemma~\ref{lem:rel1} is directly applicable to deceitful faults. This does not necessarily mean however that a protocol that tolerates a Byzantine fault and a rational player also tolerates a Byzantine fault and a deceitful fault.

  Interestingly, an analogous proof provides the same result for a protocol that tolerates instead crash and Byzantine players. We show this result in Theorem~\ref{thm2:rel1}. However, we first define $\epsilon$-$(t',t)$-immunity to combine tolerance to a number of crash and Byzantine players together:
  \begin{defn}[$\epsilon$-$(t',t)$-immunity]
    A joint strategy $\vv{\sigma}\in \mathcal{S}$  is $\epsilon$-$(t,t')$-immune if, for all sets $T$ of Byzantine players such that $T\subseteq N$ with $|T|\leq t$, all sets $T'$ of crash players such that $T'\subseteq N$, $T\cap T'=\emptyset$, all $\vv{\tau}\in \mathcal{S}_T$, all $\vv{\theta}\in \mathcal{S}_{T'}$, all strategies $\sigma_s$ of the \scheduler, and all $i\not \in T\cup T'$, we have:
    \begin{equation*}
      u_i(\vv{\sigma}_{-\{T\cup T'\}}, \vv{\tau}_T, \vv{\theta}_{T'},\sigma_s) \geq u_i(\vv{\sigma},\sigma_s)-\epsilon.
      \end{equation*}
\end{defn}
  
  \begin{thm}
    \label{thm2:rel1}
    Let $\vv{\sigma}$ be a protocol that implements consensus and is $\epsilon$-$(t',t)$-immune. Then, $\vv{\sigma}$ is also $\epsilon$-$(t,t'+t)$-crash-robust.
  \end{thm}
  % \begin{proof}
  %   The proof is analogous to that of Lemma~\ref{lem:rel1}.
  % \end{proof}
  % \ARP{remove?}Corollary~\ref{cor2:rel1} reuses the Huntsman protocol, a protocol with the greatest robustness value from the state of the art, being $\epsilon$-$(k,t)$-robust for $n>\max(\frac{3}{2}k+3t,2(k+t))$, and applies Theorem~\ref{thm:rel1} to it in order to obtain, to the best of our knowledge, the greatest value of crash-robustness.
  
% \vincent{If we restrict the failure model to crashes, then there exists a $(k+t,t)$-robust protocol with $n>\max(\frac{3}{2}k+3t,2(k+t))$. Can we fix $t=\lceil n/3 \rceil -1$? If the failures are byzantine, then we know we can find a $\epsilon$-$(k-t,t)$-robust protocol (cf. PODC submission). Can we say that with byzantine failures, there is no $\epsilon$-$(k,t)$-robust protocol, for example?}\ARP{not sure of this either}
  % \begin{cor}
  %   \label{cor2:rel1}
  %   \ARP{remove?}
  %   The Huntsman protocol is $\epsilon$-$(k+t,t)$-crash-robust for $n>\max(\frac{3}{2}k+3t,2(k+t))$
  % \end{cor}

  \subsection{From crash to Byzantine players}
\label{sec:crashbyz2}
  One may wonder if the same result listed in Lemma~\ref{lem:rel1} is true in the opposite direction, that is, whether a protocol $\vv{\sigma}$ that is $\epsilon$-$(k,t)$-crash-robust is also $\epsilon$-$\ms{fun}(k,t)$-immune where $\ms{fun}(k,t)=s$ for some $s>0$. % In general, we show in Section~\ref{sec:proc} that this is not the case: there are protocols that solve consensus and are $(k,t)$-crash-robust that are not $1$-immune.
  
  We prove in Lemma~\ref{lem:rel2} that we can construct a protocol that implements consensus and is $\epsilon$-$s$-immune based on a protocol that is $\epsilon$-$(k,t)$-crash-robust for $s=\min(k,t)$, assuming cryptography and that the protocol does not implement a $(k,t)$-crash-baiting strategy. 
  \begin{lem}
    Let $\vv{\sigma}$ be an $\epsilon$-$(k,t)$-crash-robust protocol that implements consensus without a $(k,t)$-crash-baiting strategy with respect to $\vv{\sigma}$. Then, assuming cryptography and a public-key infrastructure scheme, there is an $\epsilon$-$\min(k,t)$-immune protocol $\vv{\sigma}'$ that implements consensus.
    \label{lem:rel2}
  \end{lem}
  \begin{proof}
    We show how we create the protocol $\vv{\sigma}'$ from $\vv{\sigma}$ to tolerate the new deviations that Byzantine players can follow. We list the possible deviations of $t$ Byzantine players in a protocol $\vv{\sigma}'$:
    \begin{enumerate}
    \item\label{dev:1} Byzantines players force disagreement by sending equivocating messages.
    \item\label{dev:2} Byzantine players stop replying.
    \item\label{dev:3} Byzantine players reply only to a subset of the correct players.
    \item\label{dev:4} Byzantine players reply wrongly formatted messages.
    \item\label{dev:5} Byzantine players force non-termination by sending equivocating messages.
    \item\label{dev:6} Byzantine players force non-decision (empty decision) by sending equivocating messages.
    \end{enumerate}
    Protocol $\vv{\sigma}$ already tolerates coalitions of up to $\min(k,t)$ players following deviation \ref{dev:2}, since $t$ crash players can stop replying. For deviation \ref{dev:1}, we show that if the $k$ rational players are not enough to cause a disagreement, then $k$ Byzantine players would also not be enough. Consider instead that $k$ rational players can make a coalition big enough to cause a disagreement, then it is clear that they would cause a disagreement unless there is a $(k,t)$-crash-baiting strategy that prevents it, by definition. Therefore, since $\vv{\sigma}$ does not implement a $(k,t)$-crash-baiting strategy, if the protocol tolerates $k$ rational players trying to deviate then it also tolerates deviation~\ref{dev:1} from Byzantine players. We now show how to construct $\vv{\sigma}'$ to be robust against the rest of the deviations. To tolerate deviation \ref{dev:4}, correct players in $\vv{\sigma}'$ ignore wrongly formatted messages, converting deviation \ref{dev:4} into the same deviation as \ref{dev:2}. Now, we consider deviations \ref{dev:5} and \ref{dev:6}. Since up to $\min(k,t)$ rational players cannot force a disagreement, these players would not even deviate to not terminate or to not decide (their expected utilities from playing such strategies is less than from following $\vv{\sigma}$), however, $\min(k,t)$ Byzantine players can have a greater expected utility from such deviations.

    First, we require every player to broadcast any signed message newly delivered. This makes deviation \ref{dev:3} not a deviation anymore, since every correct player eventually verifies and delivers all messages. Also, in the event of an impasse between two partitions (that is, deviations \ref{dev:5} and \ref{dev:6}), this makes it possible for correct players to gather enough evidence of which processes are responsible for such an event, in that they signed conflicting messages. If protocol $\vv{\sigma}$ decides an empty proposal in the absence of agreement, then we construct $\vv{\sigma}'$ so that it instead repeats the protocol in a new round. This way, we make deviations \ref{dev:5} and \ref{dev:6} the same deviation.

    What is left to prove is that it is impossible for a coalition of up to $\min(k,t)$ Byzantine players to force non-termination by leading correct players into a next round sending equivocating messages. For this purpose, recall that every message sent to a non-empty subset of correct players eventually reaches all correct players since they all broadcast all signed messages they each deliver. As such, correct players are eventually able to gather two conflicting, equivocating messages from each of the deviating players, identifying such set as responsible for the attempted equivocation. Thus, we describe the final modification of $\vv{\sigma}'$ with respect to $\vv{\sigma}$: once a correct player $i$ identifies (via conflicting signed messages) a player $j$ that sent equivocating messages, $i$ ignores any message coming directly from $j$ from that moment on. Notice that this modification thus converts deviations \ref{dev:5} and \ref{dev:6} into either deviation \ref{dev:3} or \ref{dev:2}, and we already showed that $\vv{\sigma}'$ tolerates such deviations as long as the number of Byzantine players is at most $\min(k,t)$. Therefore, we have constructed $\vv{\sigma'}$ extending $\vv{\sigma}$ so that every above-shown deviation from up to $\min(k,t)$ Byzantine players converts into a deviation that $\vv{\sigma}$ already tolerates, meaning that $\vv{\sigma'}$ is $\epsilon$-$\min(k,t)$-immune.
  \end{proof}
  Lemma~\ref{lem:rel2} establishes a relation in the opposite direction from Lemma~\ref{lem:rel1}. We conjecture that this is possible to prove even without the help of cryptography. Nevertheless, we do need to restrict the protocol $\vv{\sigma}$ to be $\epsilon$-$(k,t)$-crash-robust without the help of $(k,t)$-crash-baiting strategies. This is because we can only consider $k$ rational players that behave as Byzantine faults in terms of equivocation, that is, that try to cause a disagreement. The existence of a $(k,t)$-crash-baiting strategy means that some rational players will not try to cause a disagreement, and thus we could not rule out deviation~\ref{dev:2} as a deviation that Byzantine players can follow in order to break safety. % Nevertheless, this is not surprising, since Byzantine players represent many more deviations by players whose utilities are not even defined, it is expected that more assumptions are required by the direction of Lemma~\ref{lem:rel2} than by that of Lemma~\ref{lem:rel1}. Since the proof of Lemma~\ref{lem:rel2} is constructive in providing $\vv{\sigma'}$, we refer to the resulted $\vv{\sigma'}$ as the \textit{BFT-extension} of $\vv{\sigma}$.

  %   This motivates our suggested modifications to the Huntsman protocol in Section~\ref{sec:proc} that tolerate the same value of $\epsilon$-$(k,t)$-crash-robustness as that of Corollary~\ref{cor2:rel1} in a wider range of cases than its $\epsilon$-$(k,t)$-robust counterpart.

  We show in Theorem~\ref{thm:rel2} the analogous result to Lemma~\ref{lem:rel2} as Theorem~\ref{thm:rel1} is to Lemma~\ref{lem:rel1}. The proofs of theorems~\ref{thm:rel2} and~\ref{thm4:rel2} are analogous to that of Lemma~\ref{lem:rel2}.
  \begin{thm}
    \label{thm:rel2}
    Let $\vv{\sigma}$ be an $\epsilon$-$(k,t)$-crash-robust protocol that implements consensus without a $(k,t)$-crash-baiting strategy with respect to $\vv{\sigma}$, where $k\geq t$. Then, assuming cryptography and a public-key infrastructure scheme, there is an $\epsilon$-$(k-t,t)$-robust protocol $\vv{\sigma}'$ that implements consensus.
  \end{thm}
  % \begin{proof}
  %   The proof is analogous to that of Lemma~\ref{lem:rel2}.
  % \end{proof}

  Theorem~\ref{thm:rel2} excludes protocols that make use of a $(k,t)$-crash-baiting strategy $\vv{\eta}$ to be $\epsilon$-$(k,t)$-crash-robust, we show in Theorem~\ref{thm2:rel2} that if $\vv{\eta}$ is both an efficient $(k,t)$-crash-baiting strategy and an efficient $(k,t)$-crash-baiting strategy, then there is a protocol that is $\epsilon$-$(k-t,t)$-robust.

    \begin{thm}
    \label{thm2:rel2}
    Let $\vv{\sigma}$ be an $\epsilon$-$(k,t)$-crash-robust protocol that implements consensus such that there is an efficient $(k,t)$-crash-baiting strategy $\vv{\eta}$ with respect to $\vv{\sigma}$, where $k\geq t$. Let $\vv{\sigma'}$ be the BFT-extension of $\vv{\sigma}$, assuming cryptography and a public-key infrastructure scheme. If there is $\vv{\eta'}$ such that $\vv{\eta'}$ is also an efficient $(k-t,t)$-baiting strategy with respect to $\vv{\sigma'}$, then $\vv{\sigma'}$ is an $\epsilon$-$(k-t,t)$-robust protocol that implements consensus.
  \end{thm}
  \begin{proof}
    The proof is analogous to that of Lemma~\ref{lem:rel2}, with the addition that since $\vv{\eta}'$ is an efficient $(k-t,t)$-baiting strategy with respect to $\vv{\sigma'}$, in every scenario where $\vv{\eta}$ is played in $\vv{\sigma}$, then $\vv{\eta'}$ is also played in $\vv{\sigma'}$, and since $\vv{\eta'}$ is efficient, it implements consensus.
  \end{proof}

  We make use of cryptography in Lemma~\ref{lem:rel2} in order to
offer a constructive proof that is useful for both
Theorem~\ref{thm:rel2} and Theorem~\ref{thm2:rel2}. Notice however
that it is trivial from Theorem~\ref{thm:imp1} that $k+2t<n$ and thus
$\min(k,t)<n/3$. The state of the art has already shown protocols that
are $s$-immune for $s<n/3$.

  % We show however in Theorem~\ref{thm3:rel2} a somewhat trivial proof of Theorem~\ref{thm:rel2} that is not constructive, without assuming cryptography.\ARP{shorten next theorem?}

  % \begin{thm}
  %   \label{thm3:rel2}
  %   % \ARP{shorten?}
  %   Let $\vv{\sigma}$ be an $\epsilon$-$(k,t)$-crash-robust protocol that implements consensus without a $(k,t)$-crash-baiting strategy with respect to $\vv{\sigma}$. Then, there is a $\min(k,t)$-immune protocol $\vv{\sigma}'$ that implements consensus.
  % \end{thm}
  % \begin{proof}
  %   It follows from Theorem~\ref{thm:imp1} that $k+2t<n$ and thus $\min(k,t)<n/3$. The state of the art has already shown protocols that are $s$-immune for $s<n/3$.
  % \end{proof}

  Again, notice that the results from Theorems~\ref{thm:rel2} and~\ref{thm2:rel2} assume $k\geq t$, if instead $t \geq k$, then we obtain the result from theorems~\ref{thm4:rel2}, for which we reuse the definition of $\epsilon$-$(t,t')$-immunity from Section~\ref{sec:crashbyz1}. 
  \begin{thm}
    \label{thm4:rel2}
    Let $\vv{\sigma}$ be an $\epsilon$-$(k,t)$-crash-robust protocol that implements consensus without a $(k,t)$-crash-baiting strategy with respect to $\vv{\sigma}$, where $t\geq k$. Then, assuming cryptography and a public-key infrastructure scheme, there is an $\epsilon$-$(t-k,k)$-immune protocol $\vv{\sigma}'$ that implements consensus.
  \end{thm}

\section{Conclusion \& Future Work}
\label{sec:con}
In this paper, 
we showed two interesting relations between types of faults in the solvability of the rational agreement.
First, if a consensus protocol is $\epsilon$-$(k,t)$-robust then it is also $\epsilon$-$(k+t,t)$-crash-robust, meaning that in the absence of irrational (Byzantine)
players, one can tolerate coalitions with $t$ additional rational players. 
Second, with cryptography but in the absence of a baiting strategy then we can devise a $\epsilon$-$(k-t,t)$-robust consensus protocol from 
a $\epsilon$-$(k,t)$-crash-robust consensus protocol.
%In this paper, we linked $(k,t)$-robustness to $(k,t)$-crash-robustness: if a consensus protocol $\vv{\sigma}$ is $\epsilon$-$(k,t)$-robust then it is also $\epsilon$-$(k+t,t)$-crash-robust. We also prove the opposite direction, in that if $k\geq t$, and $\vv{\sigma}$ is $\epsilon$-$(k,t)$-crash-robust without implementing a crash-baiting strategy to satisfy agreement, then we can construct a protocol $\vv{\sigma'}$ that is $\epsilon$-$(k-t,t)$-robust, assuming cryptography. 
%
%Furthermore, we show that there is an $\epsilon$-$(k,t)$-robust protocol $\vv{\sigma'}$ if instead all $(k,t)$-crash-baiting strategies of $\vv{\sigma}$ are both effective $(k,t)$-crash-baiting strategies and effective $(k,t)$-baiting strategies. 
We also prove that if a protocol is $\epsilon$-$(t',t)$-immune, then it is also $\epsilon$-$(t,t'+t)$-crash-robust, and that if a protocol $\vv{\sigma}$ is $\epsilon$-$(k,t)$-crash-robust, where $t\geq k$, then there is an $\epsilon$-$(t-k,k)$-immune protocol $\vv{\sigma}'$ that implements consensus, excluding baiting strategies. We can conclude, 
thanks to the results here outlined, that the Huntsman protocol~\cite{ranchal2021huntsman} yields the greatest crash-robustness to date under this model, as it is $(k+t,t)$-crash-robust for $n>\max(\frac{3}{2}k+3t,2(k+t))$.

Future work includes exploring values of robustness and crash-robustness in variations of this model, such as different assumptions on the communication network, and including the property of fairness for fair consensus. % probabilistically synchronous communication model, and other variations of the model into different types of faults and communication models. Additionally, we leave as future work the possibility of studying the relation between Byzantine, crash and rational players in different variations of the problem, such as that of fair consensus. % Additionally, we are working on offering a constructive proof like that of Theorem~\ref{thm:rel2} that does not need to make use of cryptography. Finally, we are also working in a protocol that implements consensus without mediators, and tolerates a greater number of rational and Byzantine players than the state of the art, thanks to implementing baiting strategies.
% \appendix

\bibliography{game-theory.bib}

\begin{thebibliography}{10}

\bibitem{abraham2019implementing}
Ittai Abraham, Danny Dolev, Ivan Geffner, and Joseph~Y. Halpern.
\newblock Implementing mediators with asynchronous cheap talk.
\newblock In {\em Proceedings of the 2019 ACM Symposium on Principles of
  Distributed Computing}, PODC '19, page 501–510, New York, NY, USA, 2019.
  Association for Computing Machinery.
\newblock \href {https://doi.org/10.1145/3293611.3331623}
  {\path{doi:10.1145/3293611.3331623}}.

\bibitem{ADGH}
Ittai Abraham, Danny Dolev, Rica Gonen, and Joe Halpern.
\newblock Distributed computing meets game theory: Robust mechanisms for
  rational secret sharing and multiparty computation.
\newblock In {\em PODC}, pages 53--62, 2006.

\bibitem{ADH19}
Ittai Abraham, Danny Dolev, and Joseph~Y. Halpern.
\newblock Distributed protocols for leader election: A game-theoretic
  perspective.
\newblock {\em ACM Trans. Econ. Comput.}, 7(1), 2019.

\bibitem{Amoussou-guenou2020}
Yackolley Amoussou-guenou, Bruno Biais, Maria Potop-butucaru, F~Paris, and Sara
  Tucci-piergiovanni.
\newblock {Rational vs Byzantine Players in Consensus-based Blockchains}.
\newblock {\em Proceedings of the 19th International Conference on Autonomous
  Agents and MultiAgent Systems}, pages 43--51, 2020.

\bibitem{bei2012distributed}
Xiaohui Bei, Wei Chen, and Jialin Zhang.
\newblock Distributed consensus resilient to both crash failures and strategic
  manipulations.
\newblock {\em arXiv preprint arXiv:1203.4324}, 2012.

\bibitem{clementi2017}
A.~{Clementi}, L.~{Gualà}, G.~{Proietti}, and G.~{Scornavacca}.
\newblock Rational fair consensus in the gossip model.
\newblock In {\em 2017 IEEE International Parallel and Distributed Processing
  Symposium (IPDPS)}, pages 163--171, 2017.
\newblock \href {https://doi.org/10.1109/IPDPS.2017.67}
  {\path{doi:10.1109/IPDPS.2017.67}}.

\bibitem{DMRS11}
Varsha Dani, Mahnush Movahedi, Yamel Rodriguez, and Jared Saia.
\newblock Scalable rational secret sharing.
\newblock In {\em PODC}, pages 187--196, 2011.

\bibitem{DLS88}
Cynthia Dwork, Nancy Lynch, and Larry Stockmeyer.
\newblock Consensus in the presence of partial synchrony.
\newblock {\em Journal of the ACM (JACM)}, 35(2):288--323, 1988.

\bibitem{ebrahimi2019getting}
Zahra Ebrahimi, Bryan Routledge, and Ariel Zetlin-Jones.
\newblock Getting blockchain incentives right.
\newblock Technical report, Tech. rep., Carnegie Mellon University Working
  Paper, 2019.

\bibitem{eclipse-attack}
Parinya Ekparinya, Vincent Gramoli, and Guillaume Jourjon.
\newblock Impact of man-in-the-middle attacks on ethereum.
\newblock In {\em 2018 IEEE 37th Symposium on Reliable Distributed Systems
  (SRDS)}, pages 11--20, 2018.
\newblock \href {https://doi.org/10.1109/SRDS.2018.00012}
  {\path{doi:10.1109/SRDS.2018.00012}}.

\bibitem{attack-of-the-clone}
Parinya Ekparinya, Vincent Gramoli, and Guillaume Jourjon.
\newblock The attack of the clones against proof-of-authority.
\newblock {\em NDSS Symposium}, 2020.

\bibitem{fischer1985impossibility}
Michael~J Fischer, Nancy~A Lynch, and Michael~S Paterson.
\newblock Impossibility of distributed consensus with one faulty process.
\newblock {\em Journal of the ACM (JACM)}, 32(2):374--382, 1985.

\bibitem{FKN10}
Georg Fuchsbauer, Jonathan Katz, and David Naccache.
\newblock Efficient rational secret sharing in standard communication networks.
\newblock In {\em Proceedings of the 7th International Conference on Theory of
  Cryptography (TCC)}, page 419?436, 2010.

\bibitem{geffner2021lower}
Ivan Geffner and Joseph~Y. Halpern.
\newblock Lower bounds implementing mediators in asynchronous systems, 2021.
\newblock \href {http://arxiv.org/abs/2104.02759} {\path{arXiv:2104.02759}}.

\bibitem{goldreichplay}
Oded Goldreich, Silvio Micali, and Avi Wigderson.
\newblock How to play any mental game.
\newblock In {\em Annual ACM Symposium on Theory of Computing}, 1987.

\bibitem{groce2012byzantine}
Adam Groce, Jonathan Katz, Aishwarya Thiruvengadam, and Vassilis Zikas.
\newblock Byzantine agreement with a rational adversary.
\newblock In {\em International Colloquium on Automata, Languages, and
  Programming}, pages 561--572. Springer, 2012.

\bibitem{HV20}
Joseph~Y. Halpern and Xavier Vila\c{c}a.
\newblock Rational consensus: Extended abstract.
\newblock In {\em Proceedings of the 2016 ACM Symposium on Principles of
  Distributed Computing}, PODC '16, page 137–146, New York, NY, USA, 2016.
  Association for Computing Machinery.
\newblock \href {https://doi.org/10.1145/2933057.2933088}
  {\path{doi:10.1145/2933057.2933088}}.

\bibitem{Halpern2020}
Joseph~Y. Halpern and Xavier Vila\c{c}a.
\newblock Rational consensus: Extended abstract.
\newblock In {\em Proceedings of the 2016 ACM Symposium on Principles of
  Distributed Computing}, PODC '16, page 137–146, New York, NY, USA, 2016.
  Association for Computing Machinery.
\newblock \href {https://doi.org/10.1145/2933057.2933088}
  {\path{doi:10.1145/2933057.2933088}}.

\bibitem{harel2020consensus}
Itay Harel, Amit Jacob-Fanani, Moshe Sulamy, and Yehuda Afek.
\newblock {Consensus in Equilibrium: Can One Against All Decide Fairly?}
\newblock In Pascal Felber, Roy Friedman, Seth Gilbert, and Avery Miller,
  editors, {\em 23rd International Conference on Principles of Distributed
  Systems (OPODIS 2019)}, volume 153 of {\em Leibniz International Proceedings
  in Informatics (LIPIcs)}, pages 20:1--20:17, Dagstuhl, Germany, 2020. Schloss
  Dagstuhl--Leibniz-Zentrum fuer Informatik.
\newblock URL: \url{https://drops.dagstuhl.de/opus/volltexte/2020/11806}, \href
  {https://doi.org/10.4230/LIPIcs.OPODIS.2019.20}
  {\path{doi:10.4230/LIPIcs.OPODIS.2019.20}}.

\bibitem{LSP82}
Leslie Lamport, Robert Shostak, and Marshall Pease.
\newblock The {B}yzantine generals problem.
\newblock {\em ACM Trans. Program. Lang. Syst.}, 4(3):382--401, July 1982.

\bibitem{lysyanskaya2006rationality}
Anna Lysyanskaya and Nikos Triandopoulos.
\newblock Rationality and adversarial behavior in multi-party computation.
\newblock In Cynthia Dwork, editor, {\em CRYPTO}, pages 180--197, 2006.

\bibitem{balance-attack}
Christopher Natoli and Vincent Gramoli.
\newblock The balance attack or why forkable blockchains are ill-suited for
  consortium.
\newblock In {\em 2017 47th Annual IEEE/IFIP International Conference on
  Dependable Systems and Networks (DSN)}, pages 579--590, 2017.
\newblock \href {https://doi.org/10.1109/DSN.2017.44}
  {\path{doi:10.1109/DSN.2017.44}}.

\bibitem{ranchal2021huntsman}
Alejandro Ranchal-Pedrosa and Vincent Gramoli.
\newblock Agreement in the presence of disagreeing rational players: The
  huntsman protocol, 2021.
\newblock \href {http://arxiv.org/abs/2105.04357} {\path{arXiv:2105.04357}}.

\bibitem{singh2009zeno}
Atul Singh, Pedro Fonseca, Petr Kuznetsov, Rodrigo Rodrigues, and Petros
  Maniatis.
\newblock Zeno: Eventually consistent byzantine-fault tolerance.
\newblock In {\em Proceedings of the 6th USENIX Symposium on Networked Systems
  Design and Implementation}, NSDI'09, page 169–184, USA, 2009. USENIX
  Association.

\bibitem{Vilaca2012}
Xavier Vila{\c{c}}a, Oksana Denysyuk, and Lu{\'i}s Rodrigues.
\newblock Asynchrony and collusion in the n-party bar transfer problem.
\newblock In Guy Even and Magn{\'u}s~M. Halld{\'o}rsson, editors, {\em
  Structural Information and Communication Complexity}, pages 183--194, Berlin,
  Heidelberg, 2012. Springer Berlin Heidelberg.

\bibitem{vilacca2011n}
Xavier Vila\c{c}a, Jo\~{a}o Leit\~{a}o, and Lu\'{\i}s Rodrigues.
\newblock N-party bar transfer: Motivation, definition, and challenges.
\newblock In {\em Proceedings of the 3rd International Workshop on Theoretical
  Aspects of Dynamic Distributed Systems}, TADDS '11, page 18–22, New York,
  NY, USA, 2011. Association for Computing Machinery.
\newblock \href {https://doi.org/10.1145/2034640.2034647}
  {\path{doi:10.1145/2034640.2034647}}.

\end{thebibliography}
\end{document}